\documentclass[american,aps,prl,reprint,superscriptaddress]{revtex4-1}
\usepackage[T1]{fontenc}
\usepackage[latin9]{inputenc}
\usepackage{amsmath}
\usepackage{amsthm}
\usepackage{amssymb}
\usepackage{graphicx}

\usepackage[normalem]{ulem}
\makeatletter
\theoremstyle{plain}
\newtheorem{thm}{\protect\theoremname}
\theoremstyle{plain}
\newtheorem{prop}[thm]{\protect\propositionname}
\theoremstyle{plain}

\theoremstyle{plain}
\newtheorem{cor}[thm]{\protect\corollaryname}

\newtheorem*{thm*}{\protect\theoremname}
\newtheorem*{prop*}{\protect\propositionname}
\newtheorem*{lem*}{\protect\lemmaname}
\newtheorem*{cor*}{\protect\corollaryname}

\ifx\proof\undefined
\newenvironment{proof}[1][\protect\proofname]{\par
\normalfont\topsep6\p@\@plus6\p@\relax
\trivlist
\itemindent\parindent
\item[\hskip\labelsep\scshape #1]\ignorespaces
}{%
\endtrivlist\@endpefalse
}
\providecommand{\proofname}{Proof}
\fi

\usepackage{babel}
\usepackage{times}
\usepackage{pxfonts}
\usepackage{color}

\definecolor{mygrey}{gray}{0.35}
\definecolor{myblue}{rgb}{0.2,0.2,0.8}
\definecolor{myzard}{cmyk}{0,0,0.05,0}
\definecolor{mywhite}{rgb}{1,1,1}
\definecolor{myred}{rgb}{0.9,0.1,0.}
\usepackage[colorlinks=true,citecolor=myblue,linkcolor=myblue,urlcolor=myblue]{hyperref}

\newcommand{\tr}{\operatorname{\bf{tr}}} 
\newcommand{\Real}{\operatorname{\bf{Re}}} 
\newcommand{\Imag}{\operatorname{\bf{Im}}} 

\renewcommand{\vec}[1]{\text{\boldmath$#1$}}

\newcommand{\bra}[1]{\langle #1|}
\newcommand{\ket}[1]{|#1\rangle}
\newcommand{\braket}[2]{\langle #1|#2\rangle}
\newcommand{\ketbra}[2]{|#1\rangle\!\langle#2|}

\makeatother

\usepackage{babel}
\providecommand{\corollaryname}{Corollary}
\providecommand{\lemmaname}{Lemma}
\providecommand{\propositionname}{Proposition}
\providecommand{\theoremname}{Theorem}

\begin{document}

\title{Stochastic Coherence Theory for Qubits}

\author{Thomas Theurer}

\affiliation{Institute of Theoretical Physics, Universität Ulm, Albert-Einstein-Allee
11, D-89069 Ulm, Germany}

\author{Alexander Streltsov}

\email{streltsov.physics@gmail.com}

\affiliation{Faculty of Applied Physics and Mathematics, Gda\'{n}sk University
of Technology, 80-233 Gda\'{n}sk, Poland}

\affiliation{National Quantum Information Centre in Gda\'{n}sk, 81-824 Sopot,
Poland}

\author{Martin B. Plenio}

\affiliation{Institute of Theoretical Physics, Universität Ulm, Albert-Einstein-Allee
11, D-89069 Ulm, Germany}

\begin{abstract}
	The resource theory of coherence studies the operational value of superpositions in quantum 
	technologies. A key question in this theory concerns the efficiency of manipulation and interconversion 
	of this resource. Here we solve this question completely for mixed states of qubits by 
	determining the optimal probabilities for mixed state conversions via stochastic incoherent 
	operations. This implies new lower bounds on the asymptotic state conversion rate between 
	mixed single-qubit states which in some cases is proven to be tight. Furthermore, we obtain 
	the minimal distillable coherence for given coherence cost among all single-qubit states, which 
	sheds new light on the irreversibility of coherence theory.
\end{abstract}
\maketitle

\textbf{\emph{Introduction.}} Quantum coherence is a fundamental feature
of quantum systems, arising from the superposition principle of quantum
mechanics~\cite{Schrodinger1935a}. This motivated the development of a rigorous resource theory of coherence~\cite{aberg2006quantifying,BaumgratzPhysRevLett.113.140401,WinterPhysRevLett.116.120404,YadinPhysRevX.6.041028,StreltsovRevModPhys.89.041003},
allowing for quantitative investigations of the role of coherence in fundamental
quantum technological applications, including quantum metrology~\cite{PhysRevLett.116.150502,MarvianPhysRevA.94.052324}, quantum algorithms~\cite{HilleryPhysRevA.93.012111,Matera2058-9565-1-1-01LT01} and quantum biology~\cite{doi:10.1080/00405000.2013.829687}.

A quantum resource theory is typically bases on two main ingredients, the free states and the free operations \cite{ChitambarPhysRevLett.117.030401,PhysRevLett.115.070503,StreltsovRevModPhys.89.041003}, both arising from additional restrictions on the set of quantum operations \cite{1367-2630-10-3-033023,brandao2008entanglement,doi:10.1142/S0217979213450197}.  
In the case of coherence theory, the free states are incoherent states, i.e., quantum states which are diagonal in a fixed basis. One reason to consider such free states is naturally given by the unavoidable interaction with the environment which leads to the destruction of coherence in the basis that defines classical states.
Regarding the free operations, several choices are discussed in the literature, leading to resource theories highlighting different aspects of coherence (see \cite{StreltsovRevModPhys.89.041003} for an overview). Here, we will focus on \emph{incoherent operations} (IO), which correspond to quantum measurements which cannot create coherence for individual measurement outcomes \cite{BaumgratzPhysRevLett.113.140401} and \emph{strictly incoherent operations} (SIO): these are quantum measurements which can neither create nor use coherence for all possible outcomes~\cite{WinterPhysRevLett.116.120404,YadinPhysRevX.6.041028}.

One of the central questions within a resource theory is the
\emph{state conversion problem}, i.e., the characterization of all
quantum states which can be created from a given state via free operations with certainty.  The answer to this questions leads to a partial order on the states which determines their usefulness or value, since a given state can be used in all protocols which require a state that can be created from it. 
The state-conversion problem within SIO and IO has been solved
for all pure states~\cite{WinterPhysRevLett.116.120404,PhysRevA.96.032316} and for mixed states of a single
qubit~\cite{ChitambarPhysRevLett.117.030401,PhysRevLett.119.140402,SciRepCoherenceTransQubit}.

A more general question concerns \emph{stochastic state conversion},
i.e., the optimal probability for incoherent transformation between
two given quantum states. For transformations between pure states,
this question has been addressed in \cite{Du:2015:CMO:2871378.2871381,Du:2017:ECM:3179543.3179552}. In this work,
we study stochastic state-conversion for general mixed states and
present a complete solution for this problem for all states of a single
qubit. Remarkably, there exists a discontinuity in the maximal probability $p(\rho\rightarrow\sigma)$ for transforming mixed $\rho$ into $\sigma$ using only incoherent operations: For fixed and mixed $\rho$, $p(\rho\rightarrow\sigma)$ is either strictly zero or takes some finite value.
From this, we will deduce that for generic states $\rho$, there exists a set of states which can neither be achieved nor approximated via stochastic incoherent operations, even with arbitrary little probability.

With the results concerning single copy transformations at hand, we are then able to give a lower bound on the asymptotic conversion rate~\cite{WinterPhysRevLett.116.120404} between qubits states, i.e. the maximal rate at which, in the limit of infinitely many copies, an initial state can be converted into the target state. This lower bound can be better than previously known bounds in~\cite{WinterPhysRevLett.116.120404} and for certain states, it coincides with upper bounds that also appeared in \cite{WinterPhysRevLett.116.120404}.
The proofs not given in the main text can be found in the Supplemental Material.

\bigskip{}

\textbf{\emph{Stochastic resource theory of coherence.}}
In this section, we lay down the foundations of this work. 
As mentioned in the introduction, a main ingredient of coherence theories are incoherent states
\begin{align}
	\rho=\sum_{i}p_{i}\ket{i}\!\bra{i}
\end{align}
which are diagonal in the fixed basis $\{\ket{i}\}$. As the free operations, we consider \emph{incoherent operations}~\cite{BaumgratzPhysRevLett.113.140401}: these are quantum
transformations $\Lambda$ which admit an incoherent Kraus decomposition 
\begin{equation}
\Lambda[\rho]=\sum_{i}K_{i}\rho K_{i}^{\dagger}\label{eq:Lambda}
\end{equation}
with incoherent Kraus operators $K_{i}$, i.e., $K_{i}\ket{m}\sim\ket{n}$
for incoherent states $\ket{m}$ and $\ket{n}$. Incoherent operations admit a natural interpretation as quantum measurements which cannot create coherence even if postselection is applied to the individual measurement outcomes $i$ identified with the Kraus operators $K_i$.
A general deterministic operation has the form~(\ref{eq:Lambda}), where the Kraus operators $K_{i}$
fulfil the completeness condition $\sum_{i}K_{i}^{\dagger}K_{i}=\openone$. To implement a stochastic incoherent operation, we formally postselect a deterministic incoherent operation according to the measurement outcomes $i$. Now assume we deal with a stochastic operation that can be decomposed into incoherent Kraus operators which are not necessarily complete, i.e., $\sum_{i}K_{i}^{\dagger}K_{i}\leq\openone$, and transforms a state $\rho$ into the state
$\rho\rightarrow\Lambda[\rho]/p$ with conversion probability $p=\mathrm{Tr}(\Lambda[\rho])$. If we want to call this operation incoherent, we have to ensure that it is part of a deterministic incoherent operation, otherwise we would simply disregard the nonfree part of the operation. 
That this is always possible has been shown in \cite{PhysRevLett.119.230401}. Therefore we call all stochastic operations that can be decomposed into incoherent Kraus operators incoherent as well. If we can implement a stochastic transformations from a state $\rho$ to a state $\sigma$ with probability $p$, we will write $\rho\rightarrow p\sigma$.

As we will see in the following, most of the analysis in this work
can be reduced to the mathematically simpler family of \emph{strictly incoherent
operations}. These are operations that can be decomposed into strictly incoherent Kraus operators $K_i$ which are defined by the property that both $K_{i}$ and $K_{i}^{\dagger}$ are incoherent~\cite{WinterPhysRevLett.116.120404,YadinPhysRevX.6.041028}. SIO can be interpreted as quantum measurements which can neither create nor use coherence even if postselection is applied to the measurement outcomes $i$ identified with $K_i$. As in the case of IO, a free completion is possible:
\begin{prop}\label{prop:compSIO} Every stochastic quantum operation that can be decomposed into strictly incoherent Kraus operators is part of a deterministic SIO.
\end{prop}

\textbf{\emph{Coherence theory on the Bloch sphere.}}
Since part of this work is concerned with qubits, we will make frequent use of the Bloch representation, which states that every qubit state $\rho$ can be represented by a subnormalized vector $\vec{r}=(r_x,r_y,r_z)$ through
\begin{align}
	\rho=\frac{1}{2} \left( \openone +\vec{r}\ \vec{\sigma}  \right), 
\end{align}
where $\vec{\sigma}$ represents a vector containing the Pauli matrices.  As done in the Equation above, we denote density operators by small Greek letters and their Bloch vectors by the respective small Latin letter. 
Throughout the following, we assume the eigenbasis of $\sigma_z$ to be incoherent. Then rotations about the z-axis of the Bloch sphere and their inverse are both in SIO and in IO, leading to an invariance of measures and transformation probabilities under these rotations. This makes it very convenient to introduce the quantity 
\begin{align}
r=\sqrt{r_x^2+r_y^2}.
\end{align}


\textbf{\emph{Single-qubit state conversion via stochastic incoherent
operations. }}
Here we present our results concerning the optimal single-qubit state conversion via IO. Our first step is to reduce the analysis of this problem to the simpler case of SIO. In order to do this, we use the following Proposition. 

\begin{prop}
\label{prop:MixingSIO}  For two states $\rho$,
$\sigma$ and a probability $p$ let there be a stochastic SIO achieving the transformation 
\begin{equation}
\rho\rightarrow p\sigma.
\end{equation}
Then, for every incoherent state $\tau$ and every $0\leq q\leq1-p$,
there exists a stochastic SIO achieving the transformation 
\begin{equation}
\rho\rightarrow p\sigma+q\tau.
\end{equation}
\end{prop}
This allows us to prove the promised Theorem.
\begin{thm}\label{thm:equivSIOandIO}
	Let $\rho$ and $\sigma$ be states of a single qubit. The following
	statements are equivalent:\\
	(1) There exists an IO converting $\rho$ into $\sigma$ with probability~$p$.\\
	(2) There exists a SIO converting $\rho$ into $\sigma$ with probability~$p$. 
\end{thm}
With this result at hand, we are ready to state our main result, using the shorthand notation from above.
\begin{thm}\label{thm:Main} 
A qubit state $\sigma$ is reachable via a stochastic SIO or IO transformation from a fixed initial qubit state $\rho$ with a given probability $p$ iff 
\begin{subequations}
		\begin{align} 
			&r^2 s_{z}^{2}+\left(1-r_{z}^{2}\right)s^2\le r^2, \label{eq:firstThmMain}\\
			&p^2s^2 \le \frac{r^2}{1+|r_z|}\left(2p-(1-|r_z|)\right) \label{eq:secondThmMain}
		\end{align}
	\end{subequations}
	holds.
	\begin{figure}
		\centering \includegraphics[width=0.9\linewidth]{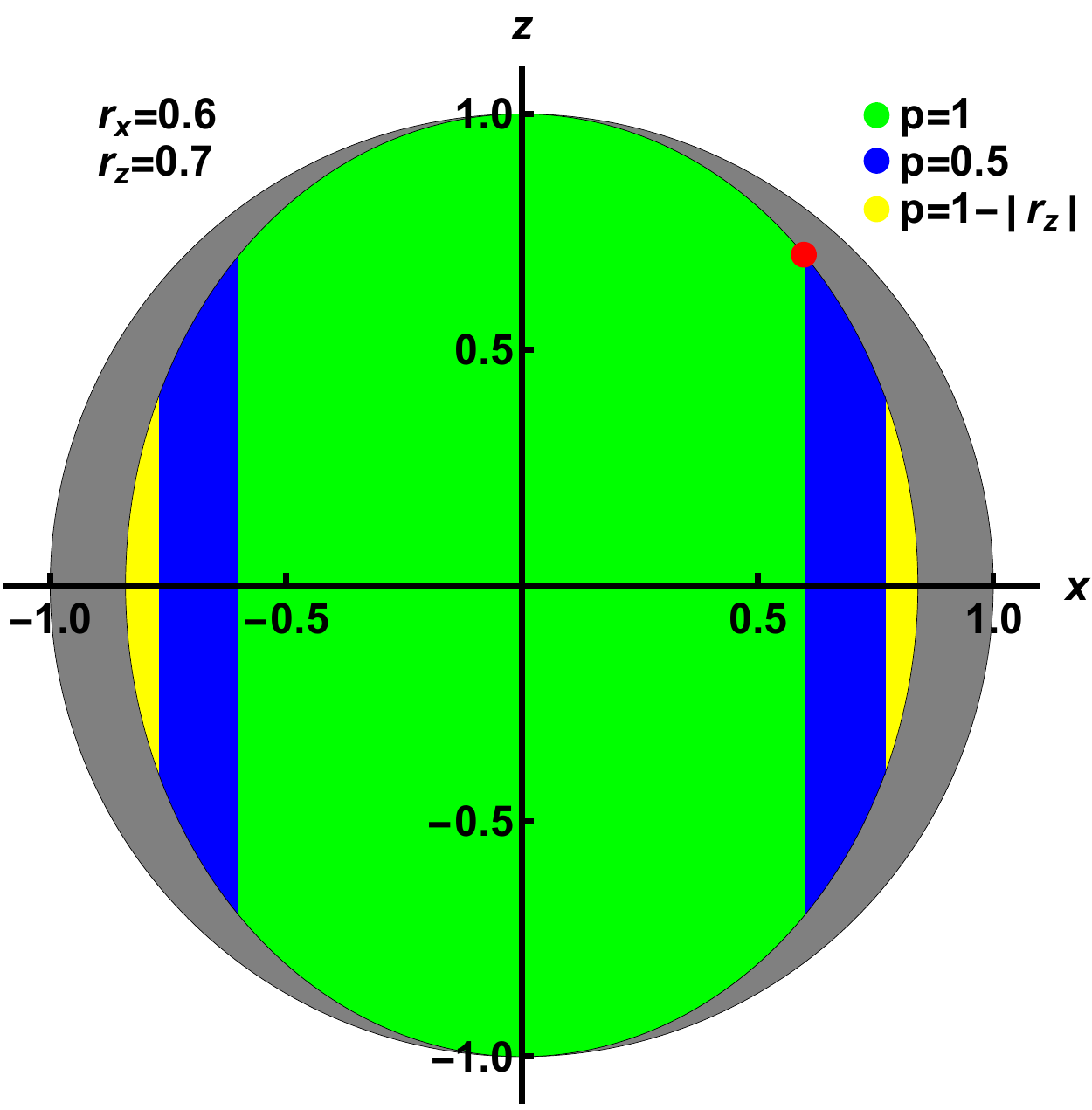} \caption{In this Figure, Thm.~\ref{thm:Main} is illustrated in the x-z plane of the Bloch sphere. For an initial state with $r_x=r=0.6$ and $r_z=0.7$, which is depicted by the red dot, the reachable regions for three different probabilities $p$ are shown. The regions which are reachable with lower probability include the ones reachable with higher probability.}
		\label{fig:ReachableRegion} 
	\end{figure}
\end{thm}

As shown in Fig.~\ref{fig:ReachableRegion}, this Theorem has a nice geometrical interpretation on the Bloch sphere: Eq.~(\ref{eq:firstThmMain}) defines an ellipsoid which is \emph{independent} of $p$ and Eq.~(\ref{eq:secondThmMain}) a cylinder which \emph{depends} on $p$. The reachable states lie inside their intersection. For $p \le 1-|r_z|$, the ellipsoid is entirely contained in the cylinder and Eq.~(\ref{eq:secondThmMain}) is automatically satisfied if Eq.~(\ref{eq:firstThmMain}) holds (see proof of Thm.~\ref{thm:Main}). Therefore, lowering the demanded probability of success below $1-|r_z|$ will not increase the set of reachable states. This implies that for mixed $\rho$, there is a discontinuity in $p(\rho\rightarrow\sigma)$ and the states outside the ellipsoid cannot be achieved via stochastic incoherent operations, even with arbitrary little probability. Since the Euclidean distance between qubits on the Bloch sphere equals twice their trace distance, this also implies that the states outside the ellipsoid cannot be approximated, because no state in a neighbourhood can be reached. 

In addition, Thm.~\ref{thm:Main} leads to the following Corollary.
\begin{cor}\label{cor:pmax}
	The maximal probability $p\left(\rho\rightarrow\sigma\right)$ 
	for a successful transformation from a coherent qubit state $\rho$
	to a coherent qubit state $\sigma$ using IO or SIO is zero if 
	\begin{align}
	r^2 s_{z}^{2}+\left(1-r_{z}^{2}\right)s^2>r^2
	\end{align}
	and 
	\begin{align}
	p(\rho\rightarrow\sigma)=\min\left\{ \frac{r^2}{\left(1+|r_z|\right)s^2}\left(1+\sqrt{1-\frac{s^2\left(1-r_z^2\right)}{r^2}}\right),1\right\} 
	\end{align}
	otherwise. 
\end{cor}

For states that are of higher dimension than two, we can give upper bounds on the maximal conversion probability via IO. Denoting by $C$ any coherence measure with the properties defined in~\cite{BaumgratzPhysRevLett.113.140401}, it holds that
\begin{align}\label{eq:boundsTrafo}
p\left(\rho\rightarrow\sigma\right)  \le \frac{C(\rho)}{C(\sigma)}.
\end{align}
Note that these bounds also appeared in \cite{Du:2015:CMO:2871378.2871381}.

\bigskip{}
\textbf{\emph{Asymptotic state conversion via IO.}} In the scenario
considered so far we assumed that incoherent operations are applied
on one copy of the state $\rho$. In the following we will extend
our investigations to asymptotic conversion scenarios, where incoherent
operations are performed on a large number of copies of the state
$\rho$. The figure of merit in this setting is the asymptotic conversion
rate 
\begin{equation}
R(\rho\rightarrow\sigma)=\sup\left\{ r:\lim_{n\rightarrow\infty}\left(\inf_{\Lambda}\left\Vert \Lambda\left(\rho^{\otimes n}\right)-\sigma^{\otimes\left\lfloor rn\right\rfloor }\right\Vert _{1}\right)=0\right\} ,\label{eq:R}
\end{equation}
where $||M||_{1}=\mathrm{Tr}\sqrt{M^{\dagger}M}$ is the trace norm,
the infimum is performed over all incoherent operations $\Lambda$,
and $\left\lfloor x\right\rfloor $ is the largest integer smaller
or equal to the real number $x$.

It is now important to note that the single copy conversion probability
$p(\rho\rightarrow\sigma)$ is a lower bound for the conversion rate:
\begin{equation}
R(\rho\rightarrow\sigma)\geq p(\rho\rightarrow\sigma).\label{eq:bound-1}
\end{equation}
In fact, asymptotic conversion at rate $p(\rho\rightarrow\sigma)$
can be achieved by applying stochastic IO on each copy of the state
$\rho$. Denoting by $C_{\mathrm{d}}$ the distillable coherence and by $C_{\mathrm{c}}$ the coherence cost~\cite{WinterPhysRevLett.116.120404}, the bounds
\begin{equation}
	\frac{C_{\mathrm{d}}(\rho)}{C_{\mathrm{c}}(\sigma)} \leq R(\rho\rightarrow\sigma)\leq\min\left\{ \frac{C_{\mathrm{d}}(\rho)}{C_{\mathrm{d}}(\sigma)},\frac{C_{\mathrm{c}}(\rho)}{C_{\mathrm{c}}(\sigma)}\right\} .\label{eq:bound-2}
\end{equation}
appeared in \cite{WinterPhysRevLett.116.120404}.

As was shown again in~\cite{WinterPhysRevLett.116.120404}, the distillable
coherence admits the following closed expression: 
\begin{equation}
C_{\mathrm{d}}(\rho)=S(\Delta[\rho])-S(\rho),
\end{equation}
where $S(\rho)=-\mathrm{Tr}[\rho\log_{2}\rho]$ is the von Neumann
entropy and $\Delta[\rho]=\sum_{i}\ketbra{i}{i}\rho\ketbra{i}{i}$
is the dephasing operator. Moreover, the coherence cost $C_{\mathrm{c}}$
is equal to the coherence of formation $C_{\mathrm{f}}$~\cite{WinterPhysRevLett.116.120404}:
\begin{equation} \label{eq:coherenceFormation}
C_{\mathrm{c}}(\rho)=C_{\mathrm{f}}(\rho)=\min\sum_{i}p_{i}S\left(\Delta\left[\psi_{i}\right]\right).
\end{equation}
Here, the minimization is performed over all pure state decompositions
of the state $\rho=\sum_{i}p_{i}\psi_{i}$. 

Up until here, the results concerning asymptotic conversions were valid for general dimensions. From here on, we will specialize them exclusively to qubits.
For single-qubit states, Eq.~(\ref{eq:coherenceFormation}) can be further simplified as follows~\cite{YuanPhysRevA.92.022124}:
\begin{equation}
C_{\mathrm{c}}(\rho)=C_{\mathrm{f}}(\rho)=h\left(\frac{1+\sqrt{1-4|\rho_{01}|^{2}}}{2}\right),\label{eq:CcQubit}
\end{equation}
where $h(x)=-x\log_{2}x-(1-x)\log_{2}(1-x)$ is the binary entropy
and $\rho_{01}=\braket{0}{\rho|1}$.

We will now demonstrate the power of these results on a specific example.
For this, we consider the following single-qubit state: 
\begin{align}
\rho & =\left(\begin{array}{cc}
\frac{2}{3} & \frac{1}{4}\\
\frac{1}{4} & \frac{1}{3}
\end{array}\right).\label{eq:ExampleRho}
\end{align}
We will study the conversion of $\rho$ into a convex combination
of maximally coherent states $\ket{\pm}=(\ket{0}\pm\ket{1})/\sqrt{2}$,
i.e., the final state $\sigma$ has the form 
\begin{equation}
\sigma=q\ket{+}\!\bra{+}+(1-q)\ket{-}\!\bra{-}.\label{eq:ExampleSigma}
\end{equation}
In Fig.~\ref{fig:bounds} we compare the aforementioned upper and
lower bounds on the state conversion rate for the states $\rho$ and
$\sigma$ in Eqs.~(\ref{eq:ExampleRho}) and~(\ref{eq:ExampleSigma}).
In particular, there exists a range of the parameter $q$ where 
the conversion probability $p(\rho\rightarrow\sigma)$ {[}solid line
in Fig.~\ref{fig:bounds}{]} is very close to the upper bound $\min\left\{ C_{\mathrm{d}}(\rho)/C_{\mathrm{d}}(\sigma),C_{\mathrm{c}}(\rho)/C_{\mathrm{c}}(\sigma)\right\} $
{[}dashed line in Fig.~\ref{fig:bounds}{]}. The true asymptotic
conversion rate $R(\rho\rightarrow\sigma)$ is between these two lines. 
The quality of our bound should also be compared to the lower bound
$C_{\mathrm{d}}(\rho)/C_{\mathrm{c}}(\sigma)$ {[}dotted line in Fig.~\ref{fig:bounds}{]}. 
The Figure clearly shows that the two different lower bounds have their advantages for
different values of the parameter $q$: For $q$ close to $1/4$, our new bound is much tighter
than the best previously known bound~\cite{WinterPhysRevLett.116.120404}. If $q$ is below a critical value, the new bound is zero. This corresponds to the region outside the reachable ellipsoid. In addition, the new bound can never exceed one, and thus the results from~\cite{WinterPhysRevLett.116.120404} give a better bound when $\sigma$
has much lower coherence than $\rho$, which corresponds to $q\approx1/2$.
\begin{figure}
\includegraphics[height=5cm]{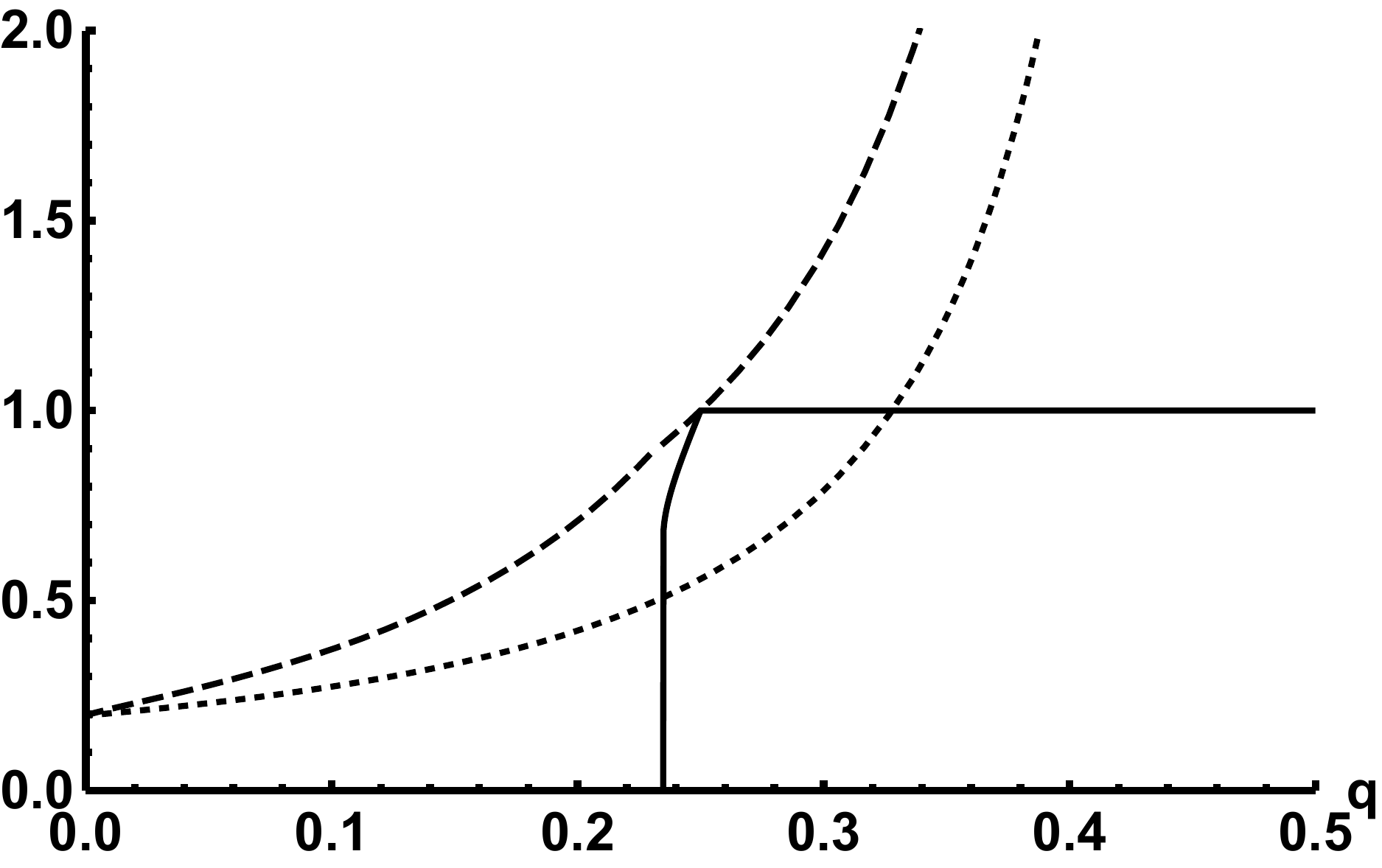}

\caption{\label{fig:bounds}Comparison of upper and lower bounds on the asymptotic
conversion rate $R(\rho\rightarrow\sigma)$ for states in Eqs.~(\ref{eq:ExampleRho})
and (\ref{eq:ExampleSigma}). Dashed line shows the upper bound given
by $\min\left\{ C_{\mathrm{d}}(\rho)/C_{\mathrm{d}}(\sigma),C_{\mathrm{c}}(\rho)/C_{\mathrm{c}}(\sigma)\right\} $,
solid line shows the lower bound given by $p(\rho\rightarrow\sigma)$,
and dotted line shows the lower bound given by $C_{\mathrm{d}}(\rho)/C_{\mathrm{c}}(\sigma)$.}

\end{figure}

Indeed, we note that for $q=1/4$ the conversion probability $p(\rho\rightarrow\sigma)$
coincides with the upper bound $C_{\mathrm{c}}(\rho)/C_{\mathrm{c}}(\sigma)$,
and in fact both are equal to $1$. This implies that the asymptotic
conversion rate is given by $R(\rho\rightarrow\sigma)=1$ in this
case. We will generalize this observation in the following Theorem,
leading to a family of single-qubit states which can be interconverted
with unit rate.
\begin{thm}\label{thm:assympUnitRate}
A state $\rho$ can be asymptotically converted into another state
$\sigma$ with optimal conversion rate $R(\rho\rightarrow\sigma)=1$
if 
\begin{equation}
s_{z}^{2}\leq r_{z}^{2}\,\,\,\mathrm{and}\,\,\,s=r.\label{eq:asymptotic}
\end{equation}
\end{thm}

\begin{figure}
\includegraphics[height=5cm]{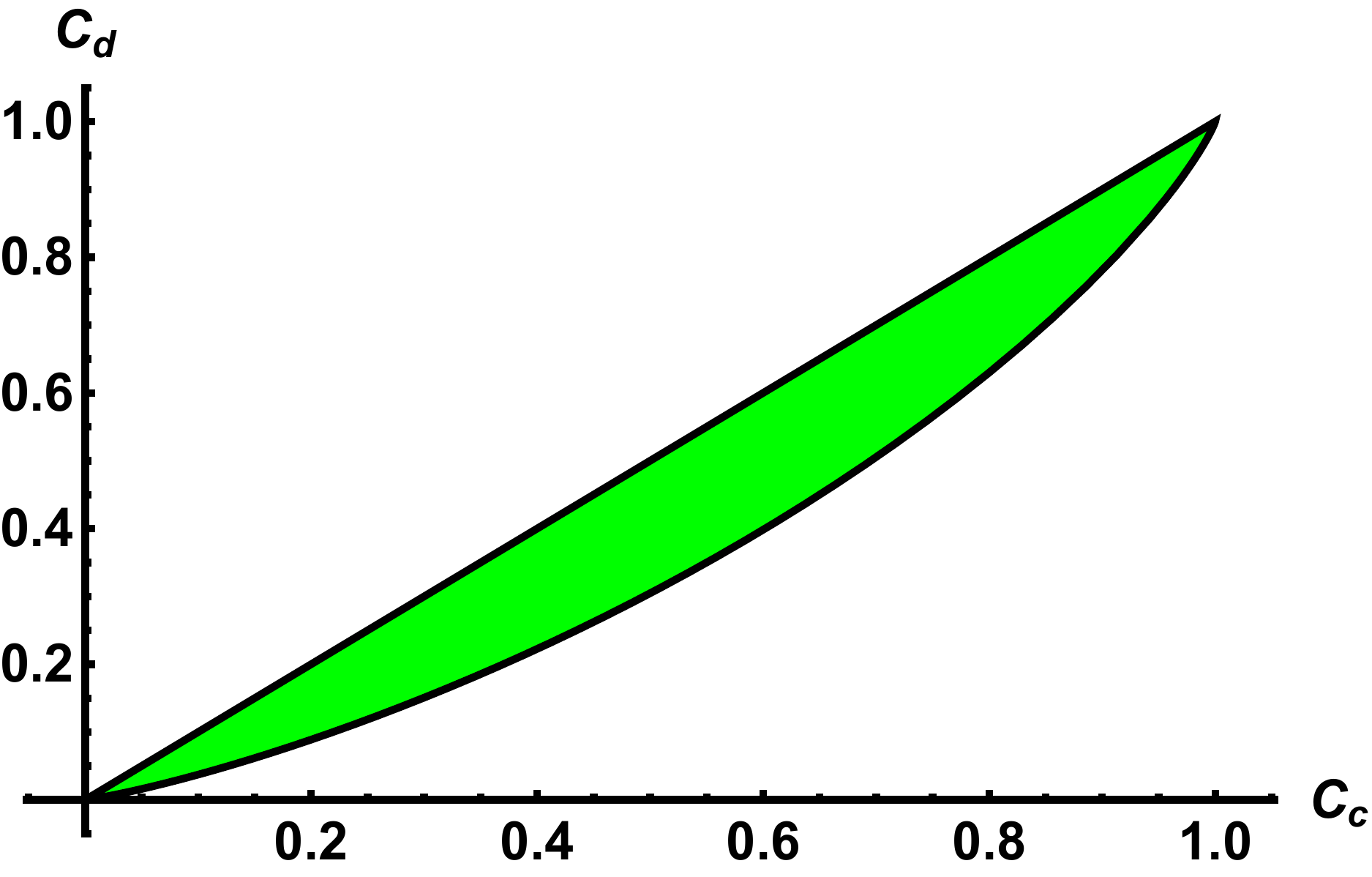}

\caption{\label{fig:irreversibility}Allowed region for distillable coherence
$C_{\mathrm{d}}$ and coherence cost $C_{\mathrm{c}}$ for single-qubit
states. The upper curve is given by $C_{\mathrm{d}}(\rho)=C_{\mathrm{c}}(\rho)$,
which is attained for pure states. The lower curve is obtained from
the family of states given in Eq.~(\ref{eq:ExampleSigma}), see main text for details.}
\end{figure}

We will now apply the methods developed in this Letter for studying
the irreversibility of coherence theory. For any quantum resource
theory, the conversion rate $R$ fulfills the following inequality
for any two nonfree states $\rho$ and $\sigma$: 
\begin{equation} \label{eq:irrevers}
R(\rho\rightarrow\sigma)\times R(\sigma\rightarrow\rho)\leq1.
\end{equation}
The resource theory is called \emph{reversible} if Eq.~(\ref{eq:irrevers})
is an equality for all nonfree states. Otherwise, the resource theory
is called \emph{irreversible}. Examples for reversible resource theories
are the theories of entanglement and coherence, when restricted to
pure states only. However, both theories are not reversible for general
mixed states~\cite{HorodeckiPhysRevLett.80.5239,WinterPhysRevLett.116.120404}.
General properties of reversible resource theories have been investigated
in~\cite{HorodeckiPhysRevLett.89.240403,PhysRevLett.115.070503}.

In the following, we will study the irreversibility of coherence theory
in more detail. In particular, we will investigate which values of
distillable coherence $C_{\mathrm{d}}$ a single-qubit state can attain,
for a fixed amount of coherence cost $C_{\mathrm{c}}$. The most interesting
family of states in this context is given by $\sigma$ in Eq.~(\ref{eq:ExampleSigma}):
this family of states has the minimal distillable coherence $C_{\mathrm{d}}$
for a fixed coherence cost $C_{\mathrm{c}}$ and vice versa maximal $C_\mathrm{c}$ for fixed $C_\mathrm{d}$~\footnote{We refer to the Supplemental Material for the proof of this statement.}.
This result allows us to plot the allowed region of coherence cost
and distillable coherence in Fig.~\ref{fig:irreversibility}. The
upper curve is given by $C_{\mathrm{d}}(\rho)=C_{\mathrm{c}}(\rho)$,
which is attained if $\rho$ is a pure state.

\bigskip{}
\textbf{\emph{Conclusions.}}
In this Letter, we studied stochastic single-qubit state conversions via incoherent operations (IO)~\cite{BaumgratzPhysRevLett.113.140401} and strictly incoherent operations (SIO)~\cite{WinterPhysRevLett.116.120404,YadinPhysRevX.6.041028}. This is an important problem, since it determines the value of qubit states for protocols using coherence. First we showed that achievable single shot conversion probabilities between qubit states are equal for SIO and IO. With the help of a recent characterization of all SIOs on qubits \cite{PhysRevLett.119.140402}, this allowed us to find simple inequalities describing all qubit states that can be reached from an initial state with fixed non-zero probability $p$. As a Corollary, we determined the maximal probability to successfully transform one qubit state into another. These results can be seen as a generalization of recent results on single-shot coherence theory~\cite{Regula1711.10512,Vijayan1804.06554}.

This single shot conversion rate gives a lower bound on
the asymptotic conversion rate, which is in some areas significantly
better than the best previously known bound~\cite{WinterPhysRevLett.116.120404}. In addition, it coincides for some states with an upper bound from \cite{WinterPhysRevLett.116.120404}, solving the asymptotic conversion problem in these cases. Finally we investigated the irreversibility of coherence theory in the asymptotic limit and determined the possible distillable coherence for fixed coherence cost.

\textbf{\emph{Acknowledgements.}} We acknowledge useful discussions with Dario Egloff and Swapan Rana.	
MBP and TT acknowledge support by the ERC Synergy Grant BioQ (grant no 319130).
AS acknowledges financial support by the National
Science Center in Poland (POLONEZ UMO-2016/21/P/ST2/04054) and the
European Union's Horizon 2020 research and innovation programme
under the Marie Sk\l{}odowska-Curie grant agreement No. 665778. 

\bibliographystyle{apsrev4-1}
\bibliography{bibliography}

\clearpage{}

\appendix

\section{Technical proofs}
Here we give the missing proofs of the main text. For readability, we restate the results.
\begin{prop*}[\ref{prop:compSIO}]\label{propA:compSIO}
	Every stochastic quantum operation that can be decomposed into strictly incoherent Kraus operators is part of a deterministic SIO.
\end{prop*}
\begin{proof} Strictly incoherent Kraus operators $K_n$ are of the form
	\begin{align}
	K_{n}=\sum_{i}c_{i,n}\ketbra{f_{n}(i)}{i}
	\end{align}
	where $f_{n}(i)$ is a bijective function on $\{1,...,d\}$. If they form a stochastic quantum operation, we have
	\begin{align}
	\sum_{n}K_{n}^{\dagger}K_{n}=\sum_{n,i}|c_{i,n}|^{2}\ketbra{i}{i}\le\openone.
	\end{align}
	Therefore $\sum_{n}|c_{i,n}|^{2}\le1\forall i$ and we can define
	\begin{align}
	\tilde{c_{i}}=\sqrt{1-\sum_{n}|c_{i,n}|^{2}}
	\end{align}
	and 
	\begin{align}
	\tilde{K}=\sum_{i}\tilde{c_{i}}\ketbra{i}{i},
	\end{align}
	which is a strictly incoherent Kraus operator and has the property 
	\begin{align}
	\tilde{K}^{\dagger}\tilde{K}+\sum_{n}K_{n}^{\dagger}K_{n}=\openone.
	\end{align}
\end{proof}

\begin{prop*}[\ref{prop:MixingSIO}]
	For two states $\rho$,
	$\sigma$ and a probability $p$ let there be a stochastic SIO achieving the transformation 
	\begin{equation}
		\rho\rightarrow p\sigma.
	\end{equation}
	Then, for every incoherent state $\tau$ and every $0\leq q\leq1-p$,
	there exists a stochastic SIO achieving the transformation 
	\begin{equation}
		\rho\rightarrow p\sigma+q\tau.
	\end{equation}
\end{prop*}
\begin{proof}
	The key idea in this proof is that the set of strictly incoherent Kraus operators is closed under concatenation. Therefore, the overall map that describes the application of a SIO on post-selected output states of another SIO is still in SIO. From Prop.~\ref{prop:compSIO} follows that we can always complete a stochastic SIO for free. The part completing the map has, with probability $1-p$, a state $\mu$ as an output. Applying total dephasing to $\mu$, we obtain an incoherent state $\mu'$, which we can transform into $\tau$ using SIO. In addition, we can do this only stochastically, which proves the Proposition. 
\end{proof} 

\begin{thm*}[\ref{thm:equivSIOandIO}]
	Let $\rho$ and $\sigma$ be states of a single qubit. The following
	statements are equivalent:\\
	(1) There exists an IO converting $\rho$ into $\sigma$ with probability~$p$.\\
	(2) There exists a SIO converting $\rho$ into $\sigma$ with probability~$p$. 
\end{thm*}

\begin{proof}
An incoherent Kraus operator $K$ is of the form
\begin{align}
	K=\sum_i c(i) \ketbra{j(i)}{i},
\end{align}
and it is strictly incoherent if $j(i)$ is one-to-one~\cite{WinterPhysRevLett.116.120404}.
Therefore all incoherent qubit Kraus operators are either also strictly incoherent or their output is, independent of the input, incoherent. Let us use the strictly incoherent ones to define a stochastic SIO. Then Prop.~\ref{prop:MixingSIO} finishes the proof. Note that this proof technique does not work in higher dimensions, since then, there exist $j(i)$ that have neither the same output for all $i$ (and have thus incoherent output), nor are they one-to-one.
\end{proof}	

\begin{thm*}[\ref{thm:Main}]\label{thmA:Main} 
	A qubit state $\sigma$ is reachable via a stochastic SIO or IO transformation from a fixed initial qubit state $\rho$ with a given probability $p$ iff 
	\begin{subequations}
		\begin{align} 
		&r^2 s_{z}^{2}+\left(1-r_{z}^{2}\right)s^2\le r^2, \label{eqA:firstThmMain}\\
		&p^2s^2 \le \frac{r^2}{1+|r_z|}\left(2p-(1-|r_z|)\right) \label{eqA:secondThmMain}
		\end{align}
	\end{subequations}
	holds.
\end{thm*}

\begin{proof}
According to Thm.~\ref{thm:equivSIOandIO}, we can focus on SIO transformations.
In order to implement a stochastic
qubit state transformation, we need a quantum instrument
with two possible outcomes, success and failure, modelled by $\mathcal{E}_{s}^{\text{SIO}}(\rho)$
and $\mathcal{E}_{f}^{\text{SIO}}(\rho)$. In the case of SIO transformations, both $\mathcal{E}_{s}^{\text{SIO}}(\rho)$
and $\mathcal{E}_{f}^{\text{SIO}}(\rho)$ have to be decomposable
into SIO Kraus operators. Due to Prop.~\ref{prop:compSIO}, we can
focus exclusively on $\mathcal{E}_{s}^{\text{SIO}}$. 
According to \cite{PhysRevLett.119.140402}, every $\mathcal{E}_{s}^{\text{SIO}}$
can be represented by four SIO Kraus operators 
\begin{align}
K_{1} & =\begin{pmatrix}a_{1} & 0\\
0 & b_{1}
\end{pmatrix},K_{2}=\begin{pmatrix}0 & b_{2}\\
a_{2} & 0
\end{pmatrix},\nonumber \\
K_{3} & =\begin{pmatrix}a_{3} & 0\\
0 & 0
\end{pmatrix},K_{4}=\begin{pmatrix}0 & b_{3}\\
0 & 0
\end{pmatrix}.
\end{align}
Since overall phases of Kraus operators are physically irrelevant,
we assume from here on $a_{i},b_{3}\ge0$. Defining $\vec{a}=(a_{1},a_{2},a_{3})$
and $\vec{b}=(b_{1},b_{2},b_{3})$, the condition that $\mathcal{E}_{s}^{\text{SIO}}$
is trace non-increasing is equivalent to $l_{a}^2:=|\vec{a}|^{2}\le1$
and $l_{b}^2:=|\vec{b}|^{2}\le1$. Due to symmetries and as explained
in \cite{PhysRevLett.119.140402}, we can restrict our analysis to
the case $s_{y}=r_{y}=0$ and $s_{x},r_{x},s_{z},r_{z}\ge0$. More
precisely, we assume $r_{x}>0$ from here on, since otherwise we have
the trivial case of incoherent initial states. From 
\begin{align}
\mathcal{E}_{s}^{\text{SIO}}(\rho)=p\sigma
\end{align}
then follow the Equations
\begin{align}
 & ps_{x}=r_{x}\left(a_{2}\Real(b_{2})+a_{1}\Real(b_{1})\right),\nonumber \\
 & 0=a_{2}\Imag(b_{2})-a_{1}\Imag(b_{1}),\nonumber \\
 & p(1+s_{z})=\left(a_{1}^{2}+a_{3}^{2}\right)(1+r_{z})+\left(|b_{2}|^{2}+b_{3}^{2}\right)(1-r_{z}),\nonumber \\
 & p(1-s_{z})=a_{2}^{2}(1+r_{z})+|b_{1}|^{2}(1-r_{z})
\end{align}
or equivalently 
\begin{alignat}{3}
 & ps_{x} &  & = &  & \ r_{x}\left(a_{2}\Real(b_{2})+a_{1}\Real(b_{1})\right),\nonumber \\
 & 0 &  & = &  & a_{2}\Imag(b_{2})-a_{1}\Imag(b_{1}),\nonumber \\
 & 2p &  & = &  & \ l_{a}^2(1+r_{z})+l_{b}^2(1-r_{z}),\nonumber \\
 & 2ps_{z} &  & = &  & \ \left(a_{1}^{2}+a_{3}^{2}-a_{2}^{2}\right)(1+r_{z})\nonumber \\
 &  &  &  &  & +\left(|b_{2}|^{2}+b_{3}^{2}-|b_{1}|^{2}\right)(1-r_{z}).
\end{alignat}
The principal idea of our proof from here on is the following: For fixed $r_{x},r_{z},p$, we determine states $(s_{x},s_{z})$ on
the boundary of the region which is achievable with stochastic SIO,
i.e. the region for which the Equations above have a solution for
suitable $\vec{a},\vec{b}$. Since the achievable region is convex
and contains the free states (we can always mix incoherently
with a free state), this will allow us to deduce the entire reachable region.

Now assume that $(s_{x},s_{z})$ is on the boundary of the reachable region. Then one can choose $a_{3}=0$
and $b_{3}=0$, since $K_{3}$ and $K_{4}$ destroy all coherence.
Formally, this can be shown considering 
\begin{align}
\vec{a}' & =(\sqrt{a_{1}^{2}+a_{3}^{2}},a_{2},0),\nonumber \\
\vec{b}' & =(|b_{1}|,\sqrt{|b_{2}|^{2}+b_{3}^{2}},0),
\end{align}
which lead to 
\begin{align}
ps_{x}'= & \ r_{x}\left(a_{2}'\Real(b_{2}')+a_{1}'\Real(b_{1}')\right)\nonumber \\
= & \ r_{x}\left(a_{2}\sqrt{|b_{2}|^{2}+b_{3}^{2}}+\sqrt{a_{1}^{2}+a_{3}^{2}}\ |b_{1}|\right)\nonumber \\
\ge & \ ps_{x},\nonumber \\
2ps_{z}'= & \ 2ps_{z},\nonumber \\
l_{a'}^2= & \ l_{a}^2,\nonumber \\
l_{b'}^2= & \ l_{b}^2,\nonumber \\
0= & a'_{2}\Imag(b'_{2})-a'_{1}\Imag(b'_{1}).
\end{align}
Remember that we consider fixed $r_{x},r_{z}$ and $p>0$. Thus $s_{x}'\ge s_{x}$
and $s_{z}'=s_{z}$. This mixing argument with the free states excludes boundaries of
the achievable region parallel to the x-axis. Therefore $s_{x}'>s_{x}$
for $s_{z}'=s_{z}$ cannot happen if both $(s_{x},s_{z})$ and $(s_{x}',s_{z}')$
lie on the boundary and we will assume from here on $a_{3}=b_{3}=0$
and $b_{1},b_{2}\ge0$. This leads to the Equations 
\begin{alignat}{3} \label{eq:reducedEqSxSzP}
 & ps_{x} &  & = &  & \ r_{x}\left(a_{2}b_{2}+a_{1}b_{1}\right),\nonumber \\
 & 2p &  & = &  & \ l_{a}^2(1+r_{z})+l_{b}^2(1-r_{z}),\nonumber \\
 & 2ps_{z} &  & = &  & \ \left(a_{1}^{2}-a_{2}^{2}\right)(1+r_{z})+\left(b_{2}^{2}-b_{1}^{2}\right)(1-r_{z}).
\end{alignat}
Next we notice that the second line in the above Equations defines
an ellipse. Remembering that we excluded the trivial case of $r_{z}=1$
by assuming $r_{x}>0$, we can therefore use the parametrization 
\begin{align}
l_{a}= & \sqrt{\frac{2p}{1+r_{z}}}\cos\left(t\right),\nonumber \\
l_{b}= & \sqrt{\frac{2p}{1-r_{z}}}\sin\left(t\right).
\end{align}
Without loss of generality, we choose $0\le t\le\pi/2$
and the condition $l_{a},l_{b}\le1$ leads to 
\begin{align} \label{eq:constraintsT}
\cos\left(t\right)\le & \sqrt{\frac{1+r_{z}}{2p}},\nonumber \\
\sin\left(t\right)\le & \sqrt{\frac{1-r_{z}}{2p}},
\end{align}
which restricts the range of $t$ further.
Next we substitute 
\begin{align}
a_{1} & =\sqrt{\frac{2p}{1+r_{z}}}\cos\left(t\right)\cos\left(\frac{\theta-\phi}{2}\right),\nonumber \\
a_{2} & =\sqrt{\frac{2p}{1+r_{z}}}\cos\left(t\right)\sin\left(\frac{\theta-\phi}{2}\right),\nonumber \\
b_{1} & =\sqrt{\frac{2p}{1-r_{z}}}\sin\left(t\right)\sin\left(\frac{\theta+\phi}{2}\right),\nonumber \\
b_{2} & =\sqrt{\frac{2p}{1-r_{z}}}\sin\left(t\right)\cos\left(\frac{\theta+\phi}{2}\right),
\end{align}
which automatically satisfies the ellipse Equation. Since all left hand sides of these Equations are positive by assumptions, we can choose without loss of generality $0\le\theta\le\pi/2$ and $-\theta\le\phi\le\theta (\Leftrightarrow 0\le \frac{\theta-\phi}{2},\frac{\theta+\phi}{2}\le \frac{\pi}{2})$.  The remaining
two Equations are then (since $p>0$) 
\begin{align} \label{eq:sxAndszByAngles}
s_{x}= & \frac{r_{x}\sin(2t)\sin(\theta)}{\sqrt{1-r_{z}^{2}}},\nonumber \\
s_{z}= & \cos(2t)\sin(\theta)\sin(\phi)+\cos(\theta)\cos(\phi).
\end{align}
When we know for every reachable $s_{x}$ the largest possible $s_{z}$, we achieved our goal of determining the boundary of the reachable region. 
Therefore we
fix $s_{x}$ and and maximize $s_{z}$. 
For fixed $s_{x}$, we obtain
from the first Equation a relation between $t$ and $\theta$, 
\begin{align}
\sin(\theta(t))=\frac{\sqrt{1-r_{z}^{2}}s_{x}}{r_{x}\sin(2t)}.
\end{align}
Using $0\le\theta\le\pi/2$, we can rewrite the second Equation as
\begin{align*}
s_{z}(t,\phi)= & \cos(2t)\sin(\theta(t))\sin(\phi)+\sqrt{1-\sin^{2}(\theta(t))}\cos(\phi),
\end{align*}
which is maximal either on the boundary or for
\begin{align*}
0= & \frac{\partial s_{z}(t,\phi)}{\partial\phi}\\
= & \sin(\theta(t))\cos(2t)\cos(\phi)-\sqrt{1-\sin^{2}(\theta(t))}\sin(\phi).
\end{align*}
Since we have $-\pi/2\le-\theta\le\phi\le\theta\le\pi/2$, this is equivalent to
\begin{align}
\phi=\arctan\left(\frac{\sin(\theta(t))\cos(2t)}{\sqrt{(1-\sin^{2}(\theta(t)))}}\right).
\end{align}
Using that $\arctan(x)$ is monotonically increasing in $x$, we find
\begin{align}
	\phi&\ge\arctan\left(\frac{-\sin(\theta)}{\sqrt{1-\sin(\theta)^2}}\right)=-\theta, \nonumber \\ \phi&\le\arctan\left(\frac{\sin(\theta)}{\sqrt{1-\sin(\theta)^2}}\right)=\theta
\end{align}
and therefore $\phi$ inside the allowed region.
Then the $s_{z}(t)$, the $s_{z}$ optimized over $\phi$, is independent
of $t$ and given by 
\begin{align}\label{eq:szMax}
s_{z}(t)=\sqrt{1-\frac{\left(1-r_{z}^{2}\right)s_{x}^{2}}{r_{x}^{2}}}.
\end{align}
Note that the expression under the square root is, due to Eq.~(\ref{eq:sxAndszByAngles}), never negative.

Now we need to check the boundaries. 
To do this, we express $t$ in terms of $\theta$ and define $X=\sin^2(\theta)$ ( therefore $(1-r_{z}^{2})s^2_{x}/r^2_{x}\le X \le 1$, again from Eq.~(\ref{eq:sxAndszByAngles})). For the moment, we assume $\cos(2t(\theta))\ge0$. This leads to
\begin{align}
	s_z^+(\phi=\theta,\theta)=&\cos(2 t(\theta)) \sin^2 (\theta)+\cos^2(\theta) \nonumber \\
	=&\sqrt{1-\frac{(1-r_z^2)s_x^2}{r_x^2 \sin^2(\theta)}} \sin^2 (\theta)+\cos^2(\theta) \nonumber \\
	=&1-X+\sqrt{1-\frac{(1-r_z^2)s_x^2}{r_x^2 X}}X \nonumber \\
	=&s_z(X).
\end{align}
Since
\begin{align}
	0=\frac{\partial}{\partial X} \left( 1-X+\sqrt{1-y/X}\ X \right)
\end{align}
has for $y\ne0$ no solutions, $s_z(X)$ attains its extrema on the boundaries. The exact maximum on the boundary depends on $t$, but it is lower than the maximum of
\begin{align}
	s_z^+(X=(1-r_{z}^{2})s^2_{x}/r^2_{x})=&1-\frac{(1-r_z^2)s_x^2}{r_x^2}, \nonumber \\
	s_z^+(X=1)=&\sqrt{1-\frac{(1-r_z^2)s_x^2}{r_x^2}}.
\end{align}
and thus smaller than the extrema inside the allowed region. In the case of $\cos(2t(\theta))\le0$, we have
\begin{align}
	s_z^-(\phi=\theta,\theta)=&\cos(2 t(\theta)) \sin^2 (\theta)+\cos^2(\theta) \nonumber \\
	=&-\sqrt{1-\frac{(1-r_z^2)s_x^2}{r_x^2 \sin^2(\theta)}} \sin^2 (\theta)+\cos^2(\theta) \nonumber \\
	=&1-X-\sqrt{1-\frac{(1-r_z^2)s_x^2}{r_x^2 X}}X \nonumber \\
	\le&s_z^+(\phi=\theta,\theta).
\end{align}
For the boundary with $\phi=-\theta$, the above considerations are the same, with the roles of $\cos(2t(\theta))\ge0$ and $\cos(2t(\theta))\le0$ inverted. 
We thus confirmed that the maximal $s_z$ for given $s_x$ is indeed given by Eq.~(\ref{eq:szMax}) and independent of $\theta$ and $t$. 

In order to finish the proof, we need to determine the reachable range of $s_x$ which depends according to Eq.~(\ref{eq:sxAndszByAngles}) on $t$ and therefore through Eqs.~(\ref{eq:constraintsT}) on $r_z$ and $p$. By the convexity of the reachable region, it is again sufficient to find the maximal reachable $s_x$. This corresponds to finding the allowed $t$ closest to $\pi/4$ (see again Eq.~(\ref{eq:sxAndszByAngles})), for which we will consider different cases. The first case is that neither of the conditions in Eq.~(\ref{eq:constraintsT}) restricts $t$, which is equivalent to 
\begin{align}
	p\le \frac{1-r_z}{2}
\end{align}
and therefore
\begin{align}
	s_x \le \frac{r_x}{\sqrt{1-r_z^2}}.
\end{align}
If
\begin{align}
	 p\le \frac{1+r_z}{2},
\end{align} the constraints are 
\begin{align}
	0\le t\le \arcsin \left(\sqrt{\frac{1-r_z}{2p}} \right).
\end{align}
For $p<1-r_z$, the upper bound on $t$ is larger than $\pi/4$, and we find the same bounds on $s_x$ as in the first case. Using
\begin{align}
	\sin\left(2 \arcsin x \right)= 2x \sqrt{1-x^2},
\end{align} we find
\begin{align}
	s_x \le \frac{r_x}{\sqrt{1+r_z}}\frac{1}{p} \sqrt{2p-(1-r_z)}
\end{align} otherwise. In the last case, for
\begin{align}
	p\ge\frac{1+r_z}{2},
\end{align}
we have a lower and an upper bound on $t$,
\begin{align}
	\arccos\left(\sqrt{\frac{1+r_z}{2 p}}\right)\le t\le \arcsin \left(\sqrt{\frac{1-r_z}{2p}} \right).
\end{align}
From Eq.~(\ref{eq:reducedEqSxSzP}), we see that the lower bound is always smaller than the upper. In addition, 
\begin{align}
	\arccos\left(\sqrt{\frac{1+r_z}{2 p}}\right) \le \arccos\left(\frac{1}{\sqrt{2}}\right)=\frac{\pi}{4}.
\end{align}
Therefore, we end up with the same conclusions as in the second case.

Finally, using the symmetry
and mixing arguments, the reachable region is defined by the inequalities 
\begin{align}
	&s_{z}^{2}\le1-\frac{1-r_{z}^{2}}{r_{x}^{2}+r_{y}^{2}}\left(s_{x}^{2}+s_{y}^{2}\right) \nonumber \\
	&\begin{cases}
	p< 1-|r_z|: &s_x^2+s_y^2 \le \frac{r_x^2+r_y^2}{1-r_z^2}  \\
	p\ge 1-|r_z|:&s_x^2+s_y^2 \le \frac{r_x^2+r_y^2}{1+|r_z|} \frac{1}{p^2}\left(2p-(1-|r_z|)\right)  \\
	\end{cases}
\end{align}
Rearranging the terms in the above Equations and using the short hand notations leads to 
\begin{align} 
&r^2 s_{z}^{2}+\left(1-r_{z}^{2}\right)s^2 \le r^2, \label{eqA:ellipse} \\
&\begin{cases}
p< 1-|r_z|: & \left(1-r_z^2\right)s^2 \le r^2,  \\
p\ge 1-|r_z|:&p^2s^2 \le \frac{r^2}{1+|r_z|}\left(2p-(1-|r_z|)\right),  \\
\end{cases}
\end{align}
formally also including the trivial cases of $r_{x}=r_{y}=0$. Now one can easily see that the condition for $p\le 1-|r_z|$ is always satisfied if condition (\ref{eqA:ellipse}) is satisfied.
If we insert $p=1-|r_z|$ into the condition for $p\ge 1-|r_z|$, we obtain after simplifications
\begin{align}
	(1-r_z^2) s^2\le r^2,
\end{align}
which is also always satisfied if condition~(\ref{eqA:ellipse}) is satisfied. Therefore the condition 
\begin{align}
	p^2s^2 \le \frac{r^2}{1+|r_z|}\left(2p-(1-|r_z|)\right)
\end{align}
is for $p\le 1-|r_z|$ automatically satisfied, if condition~(\ref{eqA:ellipse}) holds.
This leads us to the Theorem.
\end{proof}

\begin{cor*}[\ref{cor:pmax}]\label{corA:pmax}
	The maximal probability $p\left(\rho\rightarrow\sigma\right)$ 
	for a successful transformation from a coherent qubit state $\rho$
	to a coherent qubit state $\sigma$ using IO or SIO is zero if 
	\begin{align}
	r^2 s_{z}^{2}+\left(1-r_{z}^{2}\right)s^2>r^2
	\end{align}
	and 
	\begin{align}
	p(\rho\rightarrow\sigma)=\min\left\{ \frac{r^2}{\left(1+|r_z|\right)s^2}\left(1+\sqrt{1-\frac{s^2\left(1-|r_z|\right)}{r^2}}\right),1\right\} 
	\end{align}
	otherwise. 
\end{cor*}
\begin{proof} 
	From Thm.~\ref{thm:Main} and the comments below, we get that a transformation
	from $\rho$ to $\sigma$ (with $\rho$ coherent, i.e. $r>0$
	and therefore $r_{z}^{2}<1$) is possible with probability $p>0$
	iff 
	\begin{align}
	r^2 s_{z}^{2}+\left(1-r_{z}^{2}\right)s^2\le r^2.
	\end{align}
	As soon as we are inside this ellipsoid, the maximal probability of
	success is bounded by Eq.~(\ref{eqA:secondThmMain}). Now we want to maximize $p$ such that this inequality is still satisfied. This is the case if we choose the larger $p$ for which 
	\begin{align}
	p^2s^2 = \frac{r^2}{1+|r_z|}\left(2p-(1-|r_z|)\right).
	\end{align}
	Together with the assumptions that $p_{\max}$ is
	a probability, this finishes the proof. 
\end{proof}

\begin{thm*}[\ref{thm:assympUnitRate}]
	A state $\rho$ can be asymptotically converted into another state
	$\sigma$ with optimal conversion rate $R(\rho\rightarrow\sigma)=1$
	if 
	\begin{equation}\label{eq:ApAsymptotic}
	s_{z}^{2}\leq r_{z}^{2}\,\,\,\mathrm{and}\,\,\,s=r.
	\end{equation}
\end{thm*}
\begin{proof}
	In the first step of the proof note that $p(\rho\rightarrow\sigma)=1$
	for any two states $\rho$ and $\sigma$ fulfilling Eqs.~(\ref{eq:ApAsymptotic}),
	which follows directly from Eqs.~(3a) and (3b) in~\cite{PhysRevLett.119.140402}.
	This proves that $R(\rho\rightarrow\sigma)\geq1$ in this case. 
	
	In the next step we will show that states fulfilling Eqs.~(\ref{eq:ApAsymptotic})
	have equal coherence cost: 
	\begin{equation}
	C_{\mathrm{c}}(\rho)=C_{\mathrm{c}}(\sigma).\label{eq:Cc}
	\end{equation}
	Since $C_{\mathrm{c}}(\rho)/C_{\mathrm{c}}(\sigma)$ is an upper bound
	on the conversion rate, this will then complete the proof of the Theorem.
	For proving Eq.~(\ref{eq:Cc}), note that $r^2=r_{x}^{2}+r_{y}^{2}=4|\rho_{01}|^{2}$.
	Thus, Eqs.~(\ref{eq:ApAsymptotic}) directly imply the equality $|\rho_{01}|^{2}=|\sigma_{01}|^{2}$.
	Now note that for any single-qubit state $\rho$ the coherence cost
	is a simple function of $|\rho_{01}|^{2}$, see also Eq.~(\ref{eq:CcQubit}) in the main text.
	This completes the proof of Eq.~(\ref{eq:Cc}) and also the proof
	of the Theorem. 
\end{proof}

\section{Bounds on transformation probability}
Here we give the proof for the bounds in Eq.~(\ref{eq:boundsTrafo}).
Every stochastic coherence transformation from $\rho$ to $\sigma$ can be described by an incoherent quantum instrument with two possible outcomes, success and failure. We denote by $K_n$ the incoherent Kraus operators modelling the case of success and by $L_m$ the ones describing the event of failure. With
\begin{align}
p_n=&\tr\left(K_n \rho K_n^\dagger\right), \nonumber \\
q_m=&\tr\left(L_m \rho L_m^\dagger \right), \nonumber \\
\sigma_n=&K_n \rho K_n^\dagger/p_n, \nonumber \\
\chi_m=&L_m \rho L_m^\dagger/q_m, \nonumber \\
p\left(\rho\rightarrow\sigma\right)=&\sum_n p_n, \nonumber \\
q=&\sum_m q_m, 
\end{align}
we first use property (C2b), then (C3) and finally (C1) defined in \cite{BaumgratzPhysRevLett.113.140401} to arrive at
\begin{align}
C(\rho)\ge& \sum_n p_n C(\sigma_n)+ \sum_m q_m C(\chi_m) \nonumber \\
=&p\left(\rho\rightarrow\sigma\right) \sum_n \frac{p_n}{p\left(\rho\rightarrow\sigma\right)} C(\sigma_n)+ q\sum_m \frac{q_m}{q} C(\chi_m) \nonumber \\
\ge&p\left(\rho\rightarrow\sigma\right) C\left(\sum_n \frac{p_n}{p\left(\rho\rightarrow\sigma\right)} \sigma_n\right)+ qC\left(\sum_m \frac{q_m}{q} \chi_m\right) \nonumber \\
\ge&p\left(\rho\rightarrow\sigma\right) C\left(\sigma\right).
\end{align}

\section{Minimal distillable coherence for fixed coherence cost}

Here we show that the family of states 
\begin{equation}
\mu=q\ket{+}\!\bra{+}+(1-q)\ket{-}\!\bra{-}\label{eq:sigma-2}
\end{equation}
has the minimal distillable coherence for a fixed coherence cost among all single-qubit states. 

In the first step, we recall that for any single-qubit state $\rho$
the coherence cost depends only on the absolute value of the offdiagonal
element $|\rho_{01}|=|\braket{0|\rho}{1}|$, see also Eq.~(\ref{eq:CcQubit}) in
the main text. In particular, $C_\mathrm{c}$ is a strictly monotonically increasing function of $|\rho_{01}|$. Moreover, recall that $|\rho_{01}|$ is directly related
to the Euclidian distance of the state to the incoherent axis in the
Bloch space: $r_{x}^{2}+r_{y}^{2}=4|\rho_{01}|^{2}$~\footnote{Compare also the proof of Thm.~\ref{thm:assympUnitRate}.}.
This means that all states with a fixed coherence cost have the same
distance to the incoherent axis in the Bloch space. 

In the next step, we note that for any single-qubit
state $\rho$ with Bloch vector $\boldsymbol{r}=(r_x,r_y,r_z)^T$ we can introduce the state $\tilde{\rho}$ having the Bloch coordinates 
\begin{equation}
\tilde{r}_x=\sqrt{r_x^2+r_y^2},\,\,\,\,\tilde{r}_y=0,\,\,\,\,\tilde{r}_z=r_z.
\end{equation}
The state $\tilde{\rho}$ can be obtained from $\rho$ via an incoherent unitary, and thus both states have the same coherence cost and distillable coherence. In the next step, we introduce the state $\tau$ as follows:
\begin{equation}
\tau=\frac{1}{2}\tilde{\rho}+\frac{1}{2}\sigma_{x}\tilde{\rho}\sigma_{x}.
\end{equation}
Note that $\tau$ has the same distance to the incoherent axis -- and thus the
same coherence cost -- as $\rho$ and $\tilde{\rho}$, i.e., 
\begin{equation}
C_\mathrm{c} (\tau) = C_\mathrm{c} (\tilde{\rho}) = C_\mathrm{c} (\rho).
\end{equation}
Moreover, it is straightforward to
see that $\tau$ lies on the maximally coherent plane, i.e., the plane spanned by Bloch vectors corresponding to maximally coherent states. By construction, the Bloch vector of $\tau$ also lies in the $x$-$z$ plane, which implies that $\tau$ has the desired form~(\ref{eq:sigma-2}).

In the final step, recall that the distillable
coherence is convex, and thus
\begin{equation}
C_{\mathrm{d}}(\tau)\leq\frac{1}{2}C_{\mathrm{d}}(\tilde{\rho})+\frac{1}{2}C_{\mathrm{d}}(\sigma_{x}\tilde{\rho}\sigma_{x})=C_{\mathrm{d}}(\tilde{\rho})=C_{\mathrm{d}}(\rho),
\end{equation}
where we used the facts that the Pauli matrix
$\sigma_{x}$ is an incoherent unitary, and thus preserves~$C_{\mathrm{d}}$, and that $\rho$ and $\tilde{\rho}$ have the same distillable coherence. This completes the proof.

\end{document}